\documentclass[conference,10pt]{IEEEtran}
\makeatletter
\def\ps@headings{%
\def\@oddhead{\mbox{}\scriptsize\rightmark \hfil \thepage}%
\def\@evenhead{\scriptsize\thepage \hfil\leftmark\mbox{}}%
\def\@oddfoot{}%
\def\@evenfoot{}}
\makeatother
\usepackage{amssymb,amsmath}
\usepackage{paralist,color}
\usepackage{algorithm2e}

\usepackage{graphicx,cite}

\newcommand{\N}{\mathcal{N}}

\newtheorem{theorem}{\textbf{Theorem}}
\newtheorem{lemma}{\textbf{Lemma}}

\newcommand{\nix}[1]{}

\begin{document}
\title{Distributed  Flooding-based Storage Algorithms \\for Large-scale  Sensor Networks}
\author{ Salah A. Aly$^\dag$~~~~~ Moustafa Youssef$^\ddag$~~~~~Hager S. Darwish$^\dag$~~~~~Mahmoud Zidan$^\S$ \\ $^\dag$Department of Computer \& Information Science, Cairo University,  Giza,  EG 12613 \\
$^\ddag$Wireless Intelligent Networks Center (WINC), Nile University, Smart Village, EG 51034 \\
$^\S$Faculty of Computers \& Information, Ain Shams University, Cairo, EG
11566}

  \maketitle

\begin{abstract}
In this paper we propose distributed storage algorithms for
large-scale wireless sensor networks. Assume a wireless sensor
network with $n$  nodes that have limited power, memory, and
bandwidth. Each node is capable of both sensing and storing data.
Such sensor nodes might disappear from the network due to failures
or battery depletion. Hence it is desired to design efficient
schemes to collect data from these $n$ nodes. We propose two
distributed storage algorithms (DSA's)  that utilize network
flooding to solve this problem. In the first algorithm, DSA-I, we
assume that every node utilizes network flooding to disseminate its data throughout the network using a mixing time of approximately $O(n)$. We show that this algorithm is efficient in terms of
the encoding and decoding operations.   In the second
algorithm, DSA-II, we assume that the total number of nodes is not
known to every sensor; hence dissemination of the data does not
depend on $n$.  The encoding operations in this case take
$O(C\mu^2)$, where $\mu$ is the mean degree of the network graph and
$C$ is a system parameter. We evaluate the performance of the
proposed algorithms through analysis and simulation, and show that
their performance matches the  derived
theoretical results.
\end{abstract}
\section{Introduction}\label{sec:intro}
Wireless sensor networks consist of small devices (nodes)  with
limited CPU, bandwidth, and power. They can be deployed in isolated,
tragedy, and obscured fields to monitor objects, detect fires,
temperature, flood, and other disaster incidents. They can also be
used in areas difficult to reach or where it is danger for a human
being to be involved. There has been extensive research work  on
sensor networks to improve their services, power, and
operations~\cite{stojmenovic05}.  They have taken much attention
recently due to their varieties of  applications.

Assume a wireless sensor network $\N$ with $n$ nodes thrown in a
field to detect fires or to measure temperatures. Those sensors
are distributed randomly and cannot maintain routing tables or
network topology. Some nodes might disappear from the network due to
failure or battery depletion. One needs to design storage strategies
to collect sensed data from those sensors before they disappear
suddenly from the network. Such problem and their solutions have
been considered in~\cite{aly08h,aly08e,lin07a,kamra06}.

Distributed network storage codes such as Fountain codes have been used along with random walks to distribute data from a set of sources $k$ to a set of storage nodes $n \gg k$, see~\cite{dimakis06b,aly08e}.
The authors  in~\cite{aly08e,aly08h} studied a model for  distributed network storage algorithms for wireless sensor networks where $k$ sensor nodes (sources) want to disseminate their data to $n$ storage nodes with minimum computational complexity. Fountain codes and random walks in graphs are used to solve this problem, in case of the total number of sensor and storage nodes may or may not be known. In this paper we assume a  model  where all $n$ nodes in $\N$ can sense and store data. Each sensor has a buffer of total size $M$. Furthermore, every sensor can divide its buffer into $m$ slots (small buffers), each of  size $c$, i.e. $m=\lfloor M/c \rfloor$.

\medskip

 In this paper we propose a distinct model for a wireless sensor network, wherein all nodes serve as sensors/sources as well as storage/receiver nodes.
The main advantages of the proposed algorithms are as follows:
\begin{compactenum}[i)]
\item Using analysis and simulation, we show that the encoding operations, of a node to disseminate its data,  take less computational time in comparison to the previous work.
\item
One does not need to query all nodes in the network in order to retrieve information about all $n$ nodes. Only $\%20-\%30 $ of the total nodes can be queried.
\item One can query only one arbitrary node $u$ in a certain region in the network to obtain an information about this region.
\end{compactenum}

\nix{
This paper is organized as follows. In Section~\ref{sec:model} we
introduce the network model. In Sections~\ref{sec:algDSA-I,sec:analysis} we propose two storage algorithms and provide their analysis. In Section~\ref{sec:simulation} we present simulation studies of the proposed algorithms, and the paper is concluded in Section~\ref{sec:conclusion}.}

\section{Network Model and Assumptions}\label{sec:model}

In this section we present the  network model and problem definition. Consider a wireless sensor network $\N$ with $n$ sensor nodes that are uniformly distributed at random in a region $\mathcal{A}=[0,L]^2$ for some integer $L\geq 1$. The network model $\N$  can be presented by a graph $G=(V,E)$ with  a set of  nodes $V$ and a set of edges $E$. The set $V$ represents the sensors $S=\{s_1,s_2,\ldots,s_n\}$ that will measure information about a specific field. Also, $E$ represents a set of connections (links) between the sensors $S$. Two arbitrary sensors $s_i$ and $s_j$ are connected if they are in each other's transmission range.

We ensure that the network is dense, meaning with high probability
there are no isolated nodes. Let $r>0$ be a fraction. We say that
two nodes $u$ and $v$ in $V$ are connected in $G$ if and only if the
distance between them is bounded by the design parameter $r$, i.e.
$0<d(u,v) \leq r$.

\begin{figure}
\begin{center}
\includegraphics[scale=0.7]{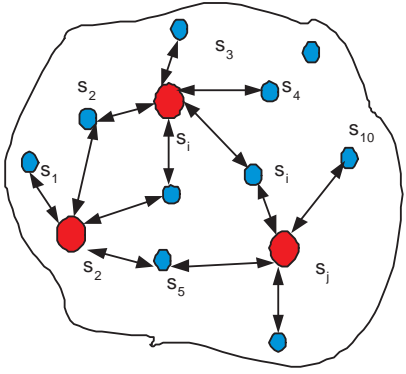}
\caption{A WSN with $n$ nodes arbitrary and randomly distributed in a field. A node $s_i$ determines its degree $d(s_i)$ by sending a flooding message to the neighboring nodes.}
\label{fig:wsn1}
\end{center}
\end{figure}

Given $u,v\in V$, we say $u$ and $v$ are \emph{adjacent}
(or $u$ is adjacent to $v$, and vice versa) if there exists a link between $u$
and $v$, i.e., $(u,v)\in E$. In this case, we also say that $u$ and $v$ are
\emph{neighbors}. Denote by $\mathcal{N}(u)$ the set of neighbors of a node
$u$. The number of neighbors, with a direct connection, of a node $u$ is called the \emph{node degree} of
$u$, and denoted by $d(u)$, i.e., $|\mathcal{N}(u)|=d(u)$. The \emph{mean
degree} of a graph $G$ is  given by
\begin{equation}\label{eq:mu}
\mu = \frac{1}{|V|}\sum_{u\in G}d(u),
\end{equation}
where $|V|$ is the total number of nodes in $G$.


The Ideal Soliton distribution $\Omega_{is}(d)$ for $k$ source blocks is given by~\cite{luby02}
\begin{equation}\label{eq:Ideal-Soliton-distribution}
\Omega_{is}(i)=\Pr(d=i)=\left\{ \begin{array}{ll} \vspace{+.05in} \displaystyle \frac{1}{k}, & i=1\\
\displaystyle\frac{1}{i(i-1)}, & i=2,3,...,k.\end{array}\right.
\end{equation}
We will use this probability distribution in the algorithms developed in the next section.

\subsection{Assumptions}

We have the following assumptions about the network model $\N$:
\begin{compactenum}[i)]
\item Let $S=\{s_1,\ldots,s_n\}$ be a set of sensing nodes that are
    distributed randomly and uniformly in a field. Each sensor acts as both a sensing and storage node. Thus, this assumption differentiate between our work
    and the problems considered in~\cite{aly08e,lin07a}.
\item Every node does not maintain routing or geographic tables, and the network topology is not known.  Every node $s_i$ can send
    a flooding message to the neighboring nodes. Also,
    every node $s_i$ can detect the total number of neighbors by broadcasting a
    simple  query message, and whoever replies to this message
    will be a neighbor of this node. Therefore, our work is more general
    and different from the work done in~\cite{dimakis07,dimakis05}. The
    degree $d(u)$ of this node is the total number of neighbors with a
    direct connection.

\begin{figure}
\begin{center}
\scalebox{0.7}{\includegraphics[width=8cm,height=3.5cm]{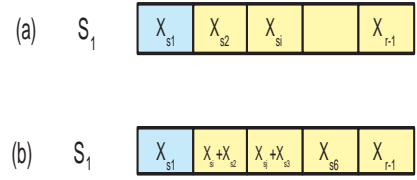}}
\caption{Every node $s_i$ has a buffer of size $M$ that  is divided into $m$ small slots. The node $s_i$ decides  with a certain probability whether to accept or reject a data $x_{s_j}$ and where to save it in one of its buffers.}
\label{fig:wsn1}
\end{center}
\end{figure}
\item  Every node has  a buffer of size $M$ and this buffer can be divided
    into smaller slots, each of size $c$, such that $m=\lfloor M/c
    \rfloor$. Hence, all nodes have the same number of slots. Also, the first slot of a node $u$ is reserved for its own sensing data.

\item Every node $s_i$ prepares a packet $packet_{s_i}$ with its $ID_{s_i}$,
    sensed data $x_{s_i}$, counter $c(x_{s_i})$, and a flag that is set to zero or one.

\item Every node draws  a degree $d_u$ from a degree distribution
    $\Omega_{is}$. If a node decided to accept a packet, it will also decide on
    which buffer it will be stored.

\end{compactenum}

\section{Distributed Storage Algorithms}\label{sec:algDSA-I}
In this section we will present a networked distributed storage algorithm
for wireless sensor networks, where all nodes act as sensing and storage nodes, and study its encoding and decoding operations.

\subsection{Encoding Operations}
We present a distributed storage  algorithm (DSA-I) for wireless
sensor networks. DSA-I algorithm consists of three main phases: Initialization,
encoding/flooding, and storage phases. Each phase can be described as
follows.

\begin{algorithm}[t!]
\SetLine%
\KwIn{A sensor network with $S=\{s_1,\ldots,s_n\}$ source nodes,  $n$ source packets $x_{s_i},\ldots,x_{s_n}$ and a positive constant $c(s_i)$.}%
\KwOut{storage buffers $y_1,y_2,\ldots,y_n$ for all sensors $S$.} %
\ForEach{node $u=1:n$}%
    {
    Generate $d_c(u)$ according to $\Omega_{is}(d)$ (or $\Omega_{rs}(d)$ and a set of neighbors $\N(u)$ using flooding.;%
    }%
\ForEach{source node $s_i, i=1:n$}%
    {
    Generate header of $x_{s_i}$ and $token=0$\;%
    Set counter $c(x_{s_i})=\lfloor n/d(s_i)\rfloor$\;%
   Flood $x_{s_i}$ to all $\mathcal{N}(s_i)$ uniformly at random, Send $x_{s_i}$ to $u \in \mathcal{N}(s_i)$ \;%

 with probability 1,  $y_u$ = $y_u\oplus x_{s_i}$\;%
 Put $x_{s_i}$ into $u$'s forward queue\;%
    $c(x_{s_i})=c(x_{s_i})-1$\;%
    }%

\While{source packets remaining}%
    {
    \ForEach{node $u$ receives packets before current round}%
        {
        Choose $v\in \mathcal{N}(u)$ uniformly at random\;%
        Send  packet $x_{s_i}$ in $u$'s forward queue to $v$\;%
        \uIf{$v$ receives $x_{s_i}$ for the first time}%
            {
            coin = rand(1)\;
                {flip a coin to accept or reject a packet} \;

            \If{$\mbox{coin}\leq \frac{1}{d_c(v)}$}%
                {
                $y_v$ = $y_v\oplus x_{s_i}$\;%
                Put $x_{s_i}$ into $v$'s forward queue\;%
            $c(x_{s_i})=c(x_{s_i})-1$%
                }%
            }%
            \uElseIf{$c(x_{s_i})\geq 1$}%
                {
                Put $x_{s_i}$ into $v$'s forward queue\;%
                $c(x_{s_i})=c(x_{s_i})-1$\;%
                }
            \Else
                {
                Discard $x_{s_i}$\;
                Hence $C(s_i)=1$ or no node to send to.
                }%
            }%
        }%

\mbox{}\\

\caption{ DSA-I Algorithm: Distributed storage algorithm  for a WSN where
the data is disseminated using multicasting and flooding to all
neighbors.} \label{alg:DSA-I}
\end{algorithm}

\subsubsection{Initialization Phase} Every node $s_i$ in $S$ has an
    $ID_{s_i}$ and sensed data $x_{s_i}$. The node $s_i$ in the
    initialization phase prepares a $packet_{s_i}$ with these values.  Also, the packet contains a hop count field, $c(x_{s_i})$, and a flag indicating whether the data is new or an update of a previous value.
    Each node will have a different hop count value depending on the number of its neighbors $d(s_i)$. Such that if a node $s_i$ has a few
    neighbors, then $c(x_{s_i})$ will be large. Also, a node with
    large number of neighbors will choose a small counter
    $c(x_{s_i})$. This means that every node will decide its own
    counter.
\begin{eqnarray}
packet_{s_i}=(ID_{s_i}, x_{s_i},c(x_{s_i}), flag)
\end{eqnarray}

The node $s_i$ broadcasts this packet to all neighboring nodes $\N(s_i)$.

\subsubsection{Encoding and Flooding Phase}
\begin{itemize}
\item After the flooding phase, every node $u$ receiving the
    $packet_{s_i}$ will check $ID_{s_i}$, accept the data $x_{s_i}$ with probability
    one, and will add this data to its buffer slots $y$.
\begin{eqnarray}
y_u^+=y_u^- \oplus x_{s_i}.
\end{eqnarray}
This is because the node $u$ is a direct neighbor of $s_i$. The data $x_{s_i}$ is disseminated rapidly to all neighbors of $s_{i}$.

\item The node $u$ will decrease the counter by one as
    \begin{eqnarray}c(x_{s_i})=c(x_{s_i})-1.\end{eqnarray}
The node $u$ will select a set of neighbors that did not
    receiver the message $x_{s_i}$ and it will unicast this message to them.
\item For an arbitrary node $v$ that receives the message from
    $u$, it will check if the $x_{s_i}$ has been received before,
    if yes, then it will discard it. If not, then it will decide whether to accept or reject
    it based on a random value drawn from $\Omega_{is}(d)$ . If
    accepted, then it will add the data to one of its buffer slots $
    y_v^+=y_v^- \oplus x_{s_i} $ and will decrease the counter
    $c(x_{s_i})=c(x_{s_i})-1.$

\item The node $v$ will check if the counter is zero, otherwise it will decrease it and send this message to the neighboring nodes that did not receive it.

\end{itemize}

\subsubsection{Storage Phase} Every node will maintain its own
buffer
    by storing a copy of its data and other nodes' data. Also, a node
    will store a list of nodes ID's of the packets that reached it.
    After all nodes receive, send, and store their own and
    neighboring data. Therefore, each node will have some information about itself and other nodes in the network.

\subsection{Decoding Operations}

The stored data can be recovered by querying a number of nodes from
the network. Let $n$ be the total number of alive nodes; assume that
every node has $m$ buffer slots such that $m=\lfloor M/c \rfloor$,
where $c$ is  a small buffer size, and $M$ is the total buffer size
in a node . In the next section, we show that  the data collector
needs to query at least $(1+\epsilon)n/m$ nodes in order to retrieve
the information about the $n$ variables.  This is much better than
previous approaches~\cite{aly08h,aly08e,lin07a} that require
querying large number of sources.

\section{DSA-I Analysis}\label{sec:analysis}
We shall provide analysis for the DSA-I algorithm shown in the previous section. The main idea is to utilize flooding and the node degree of each node to disseminate the sensed data from sensors throughout the network. We note that nodes with large degree will have  smaller counters in their packets such that their packets will travel for minimal number of neighbors. Also, nodes with smaller degree will have larger counters such that their packets will be disseminated to many neighbors as possible.
The following lemma establishes  the number of hobs (steps) that every packet will travel in the network.
\begin{lemma}\label{lem:onepacket}
On average, with a high probability,  the total number of steps for
one packet originated by a node $u$ in one branch in DSA-I is
$O(n/\mu)$.

\end{lemma}
\begin{proof}
Let $u$ be a node originating a packet $packet_u$ with degree
$d(u)$. For any arbitrary node $v$, the packet $packet_u$ will be
forwarded only if it is the first time to visit $v$ or the counter
$c(x_u)\geq 2$. We know that every packet originated from a node $u$
has a counter given by
\begin{eqnarray}
c(x_u)=\lfloor n/d(u)\rfloor.
\end{eqnarray}
  Let $\mu$ be the mean degree of the graph representing the network $\N$. On average, assuming every packet will be sent to $\mu$ neighboring
  nodes, approximating the mean degree of the graph to the degree of any
arbitrary node $u$, the result follows.
\end{proof}

If the total number of nodes is not known, one can use the method developed in~\cite{aly08e} to estimate $n$. In other words, a random walk initiated by the node $u$ can be run to estimate the total number of nodes.

\begin{lemma}\label{lem:totaltransmissions}
Let $\N$ be an instance model of a wireless sensor network with $n$
sensor nodes. The total number of transmissions required to
disseminate the information from any arbitrary node throughout the
network is $O (n)$.
\end{lemma}
\begin{proof}
Let $d(s_i)$ be the degree of a sensor node $s_i$. On average $\mu$ is the mean degree of the set of sensors $S$ approximated by $\frac{1}{n}(\sum_{i}^n d(s_i))$. Every node does flooding that takes $O(1)$ running time to $d(s_i)$ neighbors. In order to disseminate information from a sensor $s_i$, at least $n/\mu$ steps are needed using Lemma~\ref{lem:onepacket}.  Also, every sensor $s_i$ needs to send $\mu$ messages on average to the neighbors. Hence the result follows.
\end{proof}

Note that this is much better than previous results shown in~\cite{aly08e} that take $n\log n$, where $n$ is the number of sources.
\begin{theorem}
The encoding operations of DSA-I algorithm are the total number of
transmissions required to disseminate information sensed by all
nodes that is $O (n^2)$.
\end{theorem}

\section{DSA-II Algorithm Without Knowing Global Information}\label{sec:algDSA-II}
In algorithm DSA-I we assumed that the total number of nodes are
known in advance for each  sensing/storing node in the network. This
might not be the case since arbitrary  nodes might join and leave
the network at various times due to the fact that they have limited
CPU and short life time. Therefore, one needs to design a network
storage algorithm that does not depend on the value of the total
number of nodes.

We extend DSA-I to obtain a distributed storage algorithm (DSA-II)
that is totally distributed without knowing global information.  The
idea is that  each node $u$ will estimate a value for its counter
$c(u)$, the hop count, without knowing $n$.  In DSA-II each node $u$ will first perform an
inference phase that will calculate  value of the counter $c(u)$.
This can be achieved using the degree of $u$ and the degrees of the
neighboring nodes $\N(u)$. We also assume a  parameter $c_u$ that
will depend on the network condition and node's degree.


\noindent \textbf{Inference Phase:} Let $u$ be an arbitrary node in a distributed network $\N$. In the inference phase, each node $u$ will dynamically determine  value of the counter $c(u)$. The node $u$ knows its neighbors $\N(u)$. This is achieved in the flooding phase. Furthermore, the node $v$ in $\N(u)$ knows the degrees of  its neighbors.

The inference phase is done dynamically in the sense that every node
in the network will independently decide a value for its counter.
Nodes with large degrees will have a high chance of forwarding their
data throughout the network to a large number of nodes.

Let $v$ be a node connected to a source node $u$. Let $b_v$ be the
degree of a node $v$ without adding  nodes in $\N(u) \cup u$. We can
approximate the counter $c(u)$ as
\begin{eqnarray}
c(u)=c_u\Big\lfloor \frac{1}{d(u)} \sum_{v \in \N(u)} b_v \Big \rfloor
\end{eqnarray}
Once the hop counts $c(u)$ is approximated at each node $u$, the encoding operations of DSA-II algorithm are similar to encoding operations of  DSA-I algorithm.
\begin{lemma}\label{lem:totaltransmissions2}
Let $\N$ be a sensor network with $n$ sensor nodes uniformly distributed. The total number of transmissions required to disseminate the information from any arbitrary node throughout the network for the DSA-II is given by
\begin{eqnarray}
O (\mu(\mu-\lambda)),
\end{eqnarray}
where $\lambda$ be the average node density~\cite{Pe03}.
\end{lemma} 

\bigskip

\section{Practical Aspects}
In this section we shall provide evaluation and comparison analysis between DSA-I and DSA-II algorithms and related work in distributed storage algorithms. Previous work focused on utilizing random walks and Fountain codes to disseminate data sensed by a set of sensors throughout the network. Also, global and geographical information such as knowing total number of nodes, routing tables, and node locations are used.

 In this work, we  disseminate data throughout the network using data flooding once at every sensor node, then adding some redundancy at other neighboring nodes using random walks and packet trapping. Every storage node will keep track of other node's ID's, from which it will accept/reject packets.

\bigskip

The main advantages of the proposed algorithms are as follows
\begin{compactenum}[i)]
\item
One does not need to query all nodes in the network in order to
retrieve information about all $n$ nodes. Only $\%20-\%30 $ of the
total nodes can be queried.
\item One can query only one arbitrary node $u$ in a certain region in the network to obtain an information about this region.
\item
The DSA-I and DSA-II algorithms proposed in this paper are superior in comparison to the CDSA- and CDSA-II storage algorithms based on Fountain and Raptor codes proposed in~\cite{aly08e,aly08h}. The later utilize random walks to disseminate the information from a set of sources to a set of storage nodes.
\end{compactenum}
%
%

\medskip

The proposed algorithms work also in the case of data update. Assume a node $u$ sensed data $x_u$ and it has been disseminated throughout the network using flooding as shown in DSA-I and DSA-II algorithms. In this case the flag value is set to zero; and a packet from the node $u$ is originated as follows:
\begin{eqnarray}
packet_{u}=(ID_{u},x_{u},c(x_u), flag)
\end{eqnarray}
We notice that every node $v$ stores a copy from this data $x_u$ will also maintain a list of ID's including $ID_u$.
Assume $x_u^{'}$ be the new sensed data from the node $u$. The node $u$ will send update message setting the flag to one.
\begin{eqnarray}
packet_{u}=(ID_{u} ,x_u^{'}\oplus x_{u},c(x_u), flag).
\end{eqnarray}
The new and old data are Xored in this packet.
Every storage node will check the flag, whether it is an update or initial packet. Also, the node $v$ will check if $ID_u$ is in its own list. Once a node $v$ accepts the coming update packet, it will update its target buffer as
\begin{eqnarray}y_v^+=y_v^- \oplus x_u^{'}\oplus x_{u}.\end{eqnarray}

\begin{figure}
\begin{center}
\includegraphics[scale=0.44]{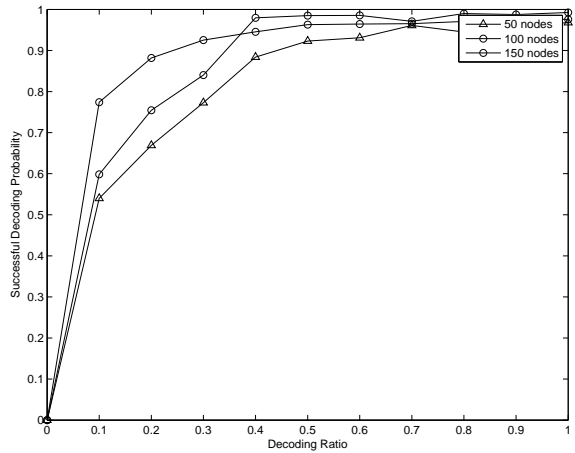}
\caption{A WSN with $n$ nodes arbitrary and randomly distributed in a field. The successful decoding ratio is shown for various values of n=50, 100, 150 with the DSA-I algorithm.}
\label{fig:DSA-I-2X2}
\end{center}
\end{figure}
\bigskip

\section{Performance and Simulation Results}\label{sec:simulation}

\bigskip

In this section we  simulate the distributed storage algorithms, DSA-I,  presented in  Section~\ref{sec:algDSA-I}. The main performance metric we investigate is the successful decoding probability versus the decoding ratio.  We define the successful decoding probability $\rho$ as percentage of $M_s$  successful trials for recovering all $n$ variables (symbols) to the total number of trails. We define $h$ to be the total number of queries needed to recover those $n$ variables. Also, we can define the decoding ratio as the total queried nodes divided by $n$, i.e. $h/n$.

We ran the experiment over a network with area $A = [0, L]^2$ grid and with different node densities. We evaluated the performance with various decoding ratios depending on the total number of nodes inside the network with incremental $\emph{step} = 0.1$.
\medskip

Fig.~\ref{fig:DSA-I-2X2} shows the decoding performance of DSA-I algorithm with Ideal Soliton distribution with small number of nodes.We ran the experiment over a network with area $A = [0, 2]^2$ grid, and evaluated the performance with various decoding ratios $0.1 \leq \eta \leq 1$. From these results we can see that the successful decoding probability increases with the gradual increases of the decoding ratio $\eta$ and reached it upper bound when $\eta = >\%30$.

\medskip

Fig.~\ref{fig:DSA-I-5X5} shows the decoding performance of DSA-I algorithm with Ideal Soliton distribution with large number of nodes.  The network is deployed in $A = [0, 5]^2$. From the simulation results we can see that the decoding ratio increases with the increase of $\lambda$ and approaches to 1 for $\eta > \%20$. Therefore the proposed algorithms perform well for large-scale wireless sensor networks.
\begin{figure}[t]
\begin{center}
\includegraphics[scale=0.44]{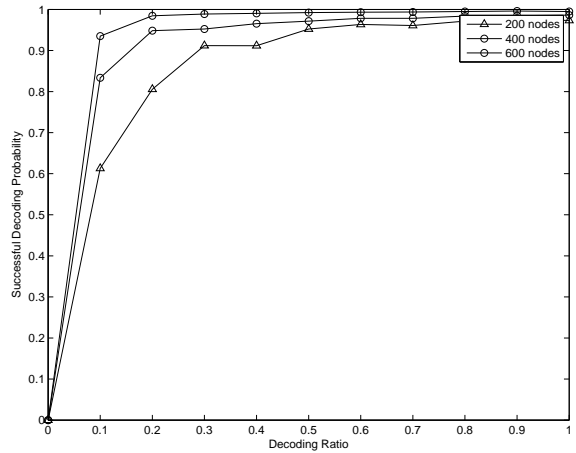}
\caption{A WSN with $n$ nodes arbitrary and randomly distributed in a field. The successful decoding ratio is shown for various values of n= 200, 400, 600 with the DSA-I algorithm.}
\label{fig:DSA-I-5X5}
\end{center}
\end{figure}

\bigskip

\section{Conclusion}\label{sec:conclusion}
We presented two distributed storage algorithms for large-scale wireless sensor networks.
Given $n$ storage/senseing  nodes, we developed schemes to disseminate sensed data throughout the network with a lesser computational overhead. The algorithms' results and performance demonstrated that it is required to query only $\%20-\%30$ of the network nodes in order to retrieve the data collected by the $n$ sensing nodes, when the buffer size is $\%10$ of the network size. Our future work will include practical and implementation aspects of these algorithms.


\scriptsize
\bibliographystyle{plain}

\bibliographystyle{ieeetr}

\end{document}